\documentclass[11pt]{article}

\usepackage{amsthm,amsfonts,amsmath,amsfonts,amsmath,amssymb}
\usepackage{xspace}
\usepackage{algorithm,algorithmic}
\usepackage[textsize=tiny,disable]{todonotes}
\usepackage[shortlabels]{enumitem}
\usepackage{epsfig}
\usepackage{wrapfig}
\usepackage{times}
\usepackage{tikz}
\usetikzlibrary{decorations.pathmorphing}
\usetikzlibrary{calc}
\usepackage{authblk}
\usepackage{url}
\usepackage{xcolor}

\DeclareMathAlphabet{\mathbbold}{U}{bbold}{m}{n}

\newtheorem{theorem}{Theorem}
\newtheorem{claim}[theorem]{Claim}

\newcommand{\PCSF}{\ensuremath{\bf{PCSF}}\xspace}
\newcommand{\PCST}{\ensuremath{\bf{PCST}}\xspace}

\newcommand{\st}{\ensuremath{\mbox{s.t.}}}
\newcommand{\opt}{\ensuremath{\texttt{opt}}}

\newcommand{\0}{\mathbbold{0}}
\newcommand{\1}{\mathbbold{1}}

\newcommand{\I}{\ensuremath{\mathcal{I}}}
\newcommand{\A}{\ensuremath{\mathcal{A}}}
\renewcommand{\Pr}{\mathbb{P}}
\newcommand{\pr}[1]{\Pr\left[#1\right]}

\newcommand{\tp}{{\scriptscriptstyle\top}}

\newcommand{\rt}{r_0}
\newcommand{\pen}[1]{\ensuremath{\pi(#1)}}

\begin{document}

\title{On the Integrality Gap of the Prize-Collecting Steiner Forest LP} 
\author[1]{Jochen K\"onemann}
\author[2]{Neil Olver}
\author[1]{Kanstantsin Pashkovich}
\author[3]{R. Ravi}
\author[1]{Chaitanya Swamy}
\author[4]{Jens Vygen}

\affil[1]{Department of Combinatorics and Optimization, University of Waterloo, 
Canada. \texttt{\{jochen,kpashkovich,cswamy\}@uwaterloo.ca}}
\affil[2]{Department of Econometrics and Operations Research, Vrije Universiteit
  Amsterdam, and CWI, Amsterdam, The Netherlands. \texttt{n.olver@vu.nl}}
\affil[3]{Tepper School of Business, Carnegie Mellon University, USA. \texttt{ravi@andrew.cmu.edu}}
\affil[4]{Research Inst. for Discrete Math., Univ. Bonn, Germany. \texttt{vygen@or.uni-bonn.de}}

\date{}

\maketitle

\begin{abstract}
  In the {\em prize-collecting Steiner forest} (\PCSF) problem, we are given an undirected
  graph $G=(V,E)$, edge costs $\{c_e\geq 0\}_{e\in E}$, terminal pairs
  $\{(s_i,t_i)\}_{i=1}^k$, and penalties $\{\pi_i\}_{i=1}^k$ for each terminal pair; the
  goal is to find a forest $F$ to minimize
  $c(F)+\sum_{i: (s_i,t_i)\text{ not connected in }F}\pi_i$. The {\em Steiner forest}
  problem can be viewed as the special case where $\pi_i=\infty$ for all $i$.
  It was widely believed that the integrality gap of the natural (and well-studied)
  linear-programming (LP) relaxation for \PCSF \eqref{lp} is at most 2.
  We dispel this belief by showing that the integrality gap of this LP is at least $9/4$. This holds even for planar graphs. We
  also show that using this LP, one cannot devise a 
  Lagrangian-multiplier-preserving (LMP) algorithm with approximation guarantee better than $4$.  
  Our results thus show a separation between the integrality gaps of the LP-relaxations
  for prize-collecting and non-prize-collecting (i.e., standard) Steiner forest, as well
  as the approximation ratios achievable relative to the optimal LP solution by LMP- and
  non-LMP-  approximation algorithms for \PCSF.   
  For the special case of {\em prize-collecting Steiner tree} (\PCST), we prove
  that the natural LP relaxation admits basic feasible solutions with all coordinates
  of value at most $1/3$ and all edge variables positive. Thus, we rule out the
  possibility of approximating \PCST with guarantee better than $3$ using a direct
  iterative rounding method. 
\end{abstract}

\section{Introduction and Background}

In an instance of the well-studied Steiner tree problem one is given
an undirected graph $G=(V,E)$, a non-negative cost $c_e$ for each edge
$e\in E$, and a set of \emph{terminals} $R \subseteq V$.  The goal is
to find a minimum-cost tree in $G$ spanning $R$. In the more general
{\em Steiner forest} problem, terminals are replaced by {\em terminal
  pairs} $(s_1,t_1), \ldots, (s_k,t_k)$ and the goal now becomes to
compute a minimum-cost forest that connects $s_i$ to $t_i$ for all
$i$. Both of the above problems are well-known to be NP- and
APX-hard~\cite{CC08,Ka72}. The best-known approximation algorithm for
the Steiner tree problem is due to Byrka et al.~\cite{BG+13} (see also
\cite{GO+12}) and achieves an approximation ratio of
$\ln 4 +\epsilon$, for any $\epsilon>0$; the Steiner forest problem admits a
$(2-1/k)$-approximation algorithm~\cite{AKR,GW}.

Our work focuses on the {\em prize-collecting} versions of the above
problems. In the prize-collecting Steiner tree problem (\PCST) we are
given a Steiner-tree instance and a non-negative penalty $\pi_v$ for
each terminal $v \in R$. The goal is to find a tree $T$ that minimizes
$c(T)+\pen{T}$, where $c(T)$ denotes the total cost of all edges in
$T$, and $\pen{T}$ denotes the total penalty of all terminals not spanned by 
$T$. In the prize-collecting Steiner forest problem (\PCSF), we are
given a Steiner-forest instance and a non-negative penalty $\pi_i$ for
each terminal pair $(s_i,t_i)$, and the goal is to find forest $F$
that minimizes $c(F)+\pen{F}$ where, similar to before, $c(F)$ is the
total cost of forest $F$, and $\pen{F}$ denotes the total penalty of
terminal pairs that are not connected by $F$. 
We can view {\PCST} as a special case of {\PCSF} by guessing a node $r$ in the optimal 
tree, 
and then modeling each vertex in $v\in R\setminus\{r\}$ by the terminal pair $(v,r)$. 

The natural integer program (IP) for \PCSF (see e.g.~\cite{BG+93}) uses a
binary variable $x_e$ for every edge $e \in E$ whose value is $1$ if
$e$ is part of the forest corresponding to $x$. The IP also has a
variable $z_i$ for each pair $(s_i,t_i)$ whose value is $1$ if $s_i$
and $t_i$ are {\em not} connected by the forest corresponding to $x$. 
We use $i \odot S$ for the predicate that is {\em
  true} if $S\subseteq V$ contains exactly one of $s_i$ and $t_i$, and
{\em false} otherwise. 
We use $\delta(S)$ to denote the set of edges with exactly one endpoint
in $S$.
In any integer solution to the LP relaxation below, 
the constraints insist that every cut separating pair 
$(s_i,t_i)$ must be crossed by the forest unless we set $z_i$ to 1 and 
pay the penalty for not connecting the terminals.
\begin{alignat}{2}
  \min & \quad & c^{\tp}x + \pi^{\tp}z & \label{lp}\tag{PCSF-LP} \\
  \st & \quad & x(\delta(S)) + z_i & \geq 1 \qquad \forall S \subseteq V,\ i \odot
          S \notag \\
  && x, z & \geq \0. \notag
\end{alignat}

Bienstock et al.~\cite{BG+93} first presented a $3$-approximation for
\PCST\ via a natural {\em threshold rounding} technique applied to this LP relaxation.
This idea also works for \PCSF, and proceeds as follows.
First, we compute a solution $(x,z)$ to the above LP. Let $R'$ be the set
of terminal pairs $(s_i,t_i)$ with $z_i < 1/3$. Note that
$\frac32 \cdot x$ is a feasible solution for the standard Steiner-forest
cut-based LP (obtained from \eqref{lp} by deleting the $z$
variables) on the instance restricted to $R'$. Thus, applying an LP-based $2$-approximation for Steiner
forest~\cite{AKR,GW} to terminal pairs $R'$ yields a forest $F'$ of cost at most
$2 \cdot \frac32 c^{\tp}x = 3c^{\tp}x$. The total penalty of the disconnected
pairs is at most $3\cdot \pi^{\tp}z$. Hence, $c(F') + \pen{F'}$ is bounded
by $3(c^{\tp}x+\pi^{\tp}z)$, and the algorithm is a $3$-approximation.
Goemans showed that by choosing a random threshold (instead of
the value $1/3$) from a suitable distribution, one can obtain an improved performance
guarantee of $1/(1-e^{-1/2})\approx 2.5415$ (see page 136 of~\cite{SWbook}, which
attributes the corresponding randomized algorithm for PCST in Section 5.7 of~\cite{SWbook}
to Goemans).

Goemans and Williamson~\cite{GW} later presented a primal-dual
$2$-approximation for \PCST{} based on the Steiner tree special case
(PCST-LP) of \eqref{lp}. In fact, the
algorithm gives even a slightly better guarantee; it produces a tree
$T$ such that 
\[ c(T) + 2\pen{T} \leq 2\cdot\opt_{\textrm{PCST-LP}}, \]
where $\opt_{\mathrm{PCST-LP}}$ is the optimum value of (PCST-LP). 
Algorithms for prize-collecting problems that achieve a performance
guarantee of the form 
\[ c(F) + \beta\cdot \pen{F} \leq \beta \cdot \opt \]
are called {\em $\beta$-Lagrangian-multiplier preserving} ($\beta$-LMP) algorithms. Such
algorithms are useful, for instance, for obtaining approximation algorithms
for the {\em partial} covering version of the problem, which in the case of Steiner tree
and Steiner forest translates to connecting at least a desired number of terminals
(e.g., see~\cite{BRV,jv01,crw04,Garg05,kps06}). Archer et 
al.~\cite{AB+11} later used the strengthened guarantee of 
Goemans and Williamson's LMP algorithm for \PCST\ to obtain a
$1.9672$-approximation algorithm for the problem.

The best known approximation guarantee for \PCSF is $2.5415$ obtained, as noted above, via
Goemans' random-threshold idea applied to the threshold-rounding algorithm of Bienstock
et al. This also shows that the integrality gap of \eqref{lp} is at most $2.5415$.
The only known lower bound prior to this work was 2.

\vspace{-1ex}
\paragraph{Our contributions.}
We demonstrate some limitations of \eqref{lp} for designing approximation algorithms for
\PCSF and its special case, \PCST, and in doing so dispel some widely-held beliefs about
\eqref{lp} and its specialization to \PCST.

The integrality gap of \eqref{lp} has been widely believed to be 2 since the work of 
Hajiaghayi and Jain~\cite{HJ}, who devised a primal-dual $3$-approximation algorithm for
\PCSF and pose the design of a primal-dual $2$-approximation based on \eqref{lp} as an
open problem.   
However, as we show here, this belief is incorrect. Our main result is as follows.

\begin{theorem}
\label{thm:pcsf-int-gap}
The integrality gap of \eqref{lp} is at least $9/4$, even for planar
instances of \PCSF. Furthermore, any
$\beta$-LMP approximation algorithm for the problem via \eqref{lp} must
have $\beta \geq 4$.
\end{theorem}

When restricted to the
non-prize-collecting Steiner forest problem, by setting
$\pi_i = \infty$ for all $i$, \eqref{lp} yields the standard LP for
Steiner forest, which has an integrality gap of 2~\cite{AKR}.  Our
result thus gives a clear separation between the integrality gaps of
the prize-collecting and standard variants.  It also shows a gap
between the approximation ratios achievable relative to
$\opt_{\text{\ref{lp}}}$ by LMP and non-LMP approximation algorithms
for \PCSF.  To the best of our knowledge, no such gaps were known
previously for an LP for a natural network design problem.  For
example, for Steiner tree, there are no such gaps relative to the
natural undirected LP obtained by specializing \eqref{lp} to \PCST.
(There are however gaps in the current best approximation ratios known
for Steiner tree and \PCST, and approximation ratios achievable for
\PCST via LMP and non-LMP algorithms.)

In order to prove Theorem \ref{thm:pcsf-int-gap} we construct an instance on a large
layered planar graph. Using a result of Carr and Vempala~\cite{CV} it follows
that \eqref{lp} has a gap of $\alpha$ iff $\alpha\cdot (x,z)$ dominates a
convex combination of integral solutions for any feasible solution $(x,z)$. We show that this can only
hold if $\alpha \geq 9/4$. 

In his groundbreaking paper \cite{Jain2001} introducing the iterative rounding method, 
Jain showed that extreme
points $x$ of the Steiner forest LP (and certain generalizations) have
an edge $e$ with $x_e = 0$ or $x_e \geq 1/2$. This then immediately
yields a $2$-approximation algorithm for the underlying
problem, by iteratively deleting an edge of value zero or rounding up an edge of value at least
half to one and proceeding on the residual instance. 
Again, it was long believed that a similar structural result holds for \PCST: extreme points
of (PCST-LP) have an edge variable of value $0$, or a variable
of value at least $1/2$. In fact, there were even stronger conjectures
that envisioned the existence of a $z$-variable with value $1$ in the
case where all edge variables had positive value less than
$1/2$. 
We refute these conjectures.

\begin{theorem}\label{thm:PCST}
  There exists an instance of \PCST where (PCST-LP) has an extreme point 
  with all edge variables positive and all variables having value at most $1/3$.
\end{theorem}

In~\cite{HN2010} it was shown, that for every vertex $(x,z)$ of
\eqref{lp} (and hence also (PCST-LP)) where $x$ is positive, there is
at least one variable of value at least $1/3$. Moreover for \eqref{lp}
this result is tight, i.e. there are instances of \PCSF such
that for some vertex $(x,z)$ of \eqref{lp}, we have $x>\0$ and all coordinates are at most 
$1/3$. However, no such example was known for (PCST-LP). 

We provide such an example for \PCST, showing that the $1/3$ upper bound on variable
values is tight also for (PCST-LP).

\section{The Integrality Gap for \PCSF}

\subsection{Lower Bound on the Integrality Gap}

We start proving Theorem~\ref{thm:pcsf-int-gap}
by describing the graph for our instance.  Let $P$ be a planar
$n$-node $3$-regular $3$-edge-connected graph (for some large enough $n$ to
be determined later). Note that such graphs exist for arbitrarily large $n$; 
e.g., the graphs of simple 3-dimensional polytopes (such as planar duals of triangulations of a sphere)
have these properties; they are 3-connected by Steinitz's theorem \cite{Steinitz}. 

We obtain $H$ from $P$ by subdividing every
edge $e$ of $P$, so that $e$ is replaced by a corresponding path with $n$ internal
nodes. Let $r$ denote an arbitrary degree-$3$ node in $H$, and call it the root.
Define
$H^{(0)}:=H$ and obtain $H^{(i)}$ from $H^{(i-1)}$ by attaching a copy
of $H$ to each degree-$2$ node $v$ in $H^{(i-1)}$, identifying the root
node of the copy with~$v$; we call this the \emph{copy of $H$ with root $v$}. 
We also define the \emph{parent} of any node $u \neq v$ in this copy to be $v$.
In the end, we let $G:=H^{(k)}$ for some
large $k$, and we let
$\rt$ be the node corresponding to the root of $H^{(0)}$.  
Figure~\ref{fig:peterson} gives an example of this construction. 
Note that each copy of $H$ can be thought of as a subgraph of $G$.

\begin{figure}
    \centering
    \includegraphics[page=1,scale=1.3]{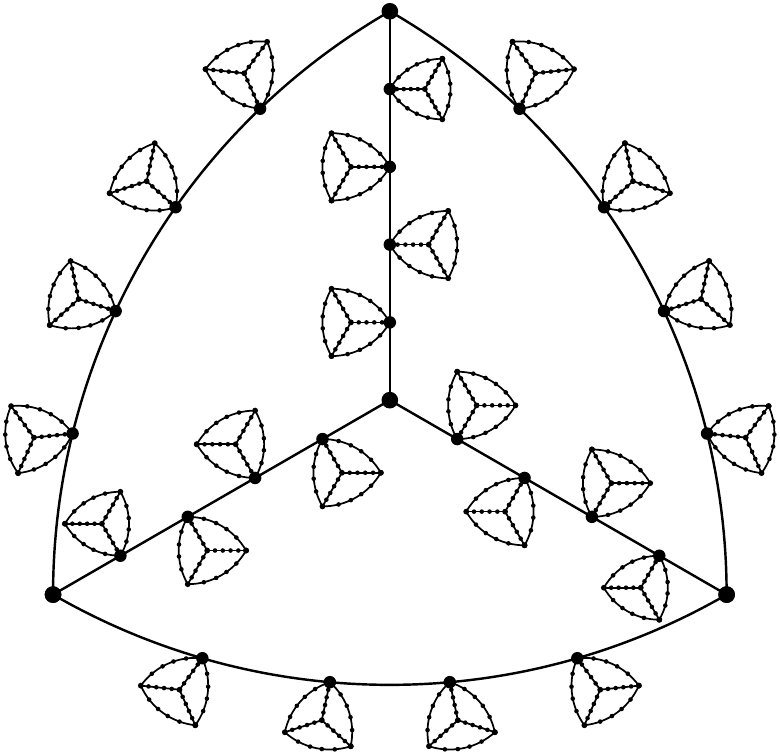}
    \caption{Taking $n=4$, and hence $P$ to be the complete graph on $4$ vertices, the resulting graph $H^{(1)}$ is shown.}
   \label{fig:peterson}
\end{figure}

Next, let us define the source-sink pairs. We introduce a
source-sink pair $s,t$ whenever $s$ and $t$ are degree-$3$ nodes in
the same copy of $H$. We also introduce a source-sink pair $\rt,t$
whenever $t$ is a degree-$2$ node in $G$.

Now let $x_e:=1/3$, for all $e \in E$, $z_{uv}:=0$ if $u$ and $v$ are degree-$3$ nodes
in the same copy of $H$, and $z_{uv}:=1/3$ otherwise. 
(Here and henceforth, we abuse notation slightly and index $z$ by the source-sink terminal pair that it corresponds to.)
Clearly, $(x,z)$
is a feasible solution for~\eqref{lp} by the 3-edge-connectivity of $G$.  

Let $\alpha$ be the integrality gap of~\eqref{lp}. By~\cite{CV} there
is a collection of forests $F_1,\ldots,F_q$ in $G$ (the same forest could appear multiple
times in the collection) such that picking a forest $F$ uniformly at random from
$F_1,\ldots,F_q$ satisfies 

\begin{enumerate}[(a)]
  \item  \label{edge_usage} $\pr{e \in F} \leq \frac{\alpha}{3}$ for all $e \in E$, and
  \item \label{connectivity} Letting $u \sim_F v$ denote the event that $u$ and $v$ are
      connected in $F$, for all $u,v \in V(G)$, we have  
      $\pr{u \sim_F v} \geq (1-\alpha z_{uv}) = \begin{cases} 
      1 & \mbox{if $u$, $v$ are degree-$3$ nodes in the same copy of $H$} \\
      1- \frac{\alpha}{3} & \mbox{if $u=\rt$ and $v$ is a degree-$2$ node in $G$} \end{cases}$
      
\end{enumerate}

We begin by observing that we may assume that each forest $F_1, \ldots, F_q$ induces a tree when restricted to any of the copies of $H$ in $G$.
For consider any $F_i$, and a copy of $H$ with root $v$; call this $H'$.
Every degree-$3$ node in $H'$ is connected to $v$ in $F_i$, by requirement (b).
So consider any degree-$2$ node $u$ in $H'$.
If $u$ is not connected to a degree-$3$ node of $H'$ (and hence to $v$) in $F_i$, 
then any edges of $F_i$ adjacent to $u$ can be safely deleted without destroying any connectivity amongst the source-sink pairs of the instance.

The argument will show that if $\alpha$ is too small, not all degree-$2$ nodes can be connected to $\rt$ with high enough probability. 
More precisely, we will show a geometrically decreasing probability, in $k$.
The intuition is roughly as follows.
Consider a copy $H'$ of $H$ with root $u$, where $u \neq \rt$.
Almost all of the degree-$2$ nodes of $H'$ that are connected to $u$ in $F$ will have degree $2$ in $F$, since $F[H']$ is a tree and $H'$ is made up of long paths.
This is rather wasteful, since both edges adjacent to a typical degree-$2$ node $v$ are used to connect; as each edge appears with probability $\alpha/3$, $v$ can only be part of $F$ (and hence connected to $u$) with probability about $\alpha/3$.
Moreover, we will show that even conditioned on the event that $u$ is \emph{not} connected
to $\rt$, there will be some choice of $v$ such that $v$ is connected to $u$ in $F$ with
probability around $2/3$ (see \eqref{eq1} in Claim~\ref{cl:condit_connect}).
This is again a waste in terms of connectivity to $\rt$.
If $p_i$ denotes the worst connectivity probability amongst nodes in $H^{(i)}$ in the construction, we have
\begin{equation*}
p_{i+1} \lessapprox \tfrac{\alpha}{3} - \tfrac23(1-p_i). 
\tag{\text{\eqref{eq5} is a more precise version of this inequality}} 
\end{equation*}
If $\alpha < 9/4$, this decreases geometrically, providing a counterexample for $n$ large enough.

For now, let us introduce an
abstract event $I$ 
(that the reader may think of as ``an ancestor of node $v$ is not
connected to $\rt$'' motivated by the above discussion).

\begin{claim}\label{cl:condit_connect}
  Let a forest $F$ be picked uniformly at random from
  $F_1,\ldots,F_q$, let $I$ be an event with $\pr{I}>0$ and let $H'$
  be a copy of $H$ in $G$. Then there exists a degree-$2$ node $v$ in $H'$
  such that
\begin{equation}
\label{eq0}
\pr{\deg_{F[H']}(v)=1} \le \frac{2}{n}
\end{equation}
and
\begin{equation}
\label{eq1}
\pr{Q_v\subseteq F \,|\, I}\ge\frac{2(n-1)}{3n}\,,
\end{equation}
where $Q_v$ is the path in $H'$ corresponding to the edge of $P$ containing $v$.

\end{claim}
\begin{proof}
  The event $I$ corresponds to a nonempty multiset $\mathcal{F}\subseteq\{F_1,\ldots,F_q\}$ of the forests.
  Each of $F_1[H']$, \ldots, $F_q[H']$ is a tree, by our earlier assumption, and so each of them naturally induce a spanning tree of $P$.
  More precisely, for each $e \in E(P)$, let $Q_e$ denote the corresponding path in $H'$; then $\{ e \in E(P): Q_e \subseteq F_i[H']\}$ is a spanning tree for each $i$.
  Thus 
  \[\sum_{e\in E(P)} |\{F'\in\mathcal{F} : Q_e\subseteq F' \}| 
  = \sum_{F'\in\mathcal{F}} |\{e\in E(P) : Q_e\subseteq F' \}| 
  = \sum_{F'\in\mathcal{F}} (n-1) 
  = |\mathcal{F}| (n-1).\]
  So there is an edge $f\in E(P)$ for which 
  \[ \pr{Q_f\subseteq F \,|\, I} = \frac{|\{F'\in\mathcal{F} : Q_f\subseteq F' \}|}{|\mathcal{F}|} \ge \frac{(n-1)}{|E(P)|} = \frac{2(n-1)}{3n}. \]

  At most two of the nodes on $Q_f$ are leaves in any
  of $F_1[H']$, \ldots, $F_q[H']$ (again since they are all trees).
  The total number of
  degree-$2$ nodes in $H'$ lying on $Q_f$ is $n$, so there exists a
  degree-$2$ node $v$ in $H'$ such that $v\in Q_f$ and
  $\pr{\deg_{F[H']}(v)=1} \le \frac{2}{n}$.
 \end{proof}

\begin{claim}\label{cl:not_connected_node}
    Let $\epsilon > 0$ be given. Then for $n$ and $k$ chosen sufficiently large,
  there exists a degree-$2$ node $u$ in $G$ such
  that
  \[
  \pr{u\sim_F \rt} \le \alpha -2+\epsilon\,,
  \]
  where $F$ is a uniformly random forest from $F_1, \ldots, F_q$.
\end{claim}
\begin{proof}
  Consider the root copy $H^{(0)}$ of $H$, with root $\rt$. 
  Set $H_0=H^{(0)}$.
  Pick a degree-$2$ node $v$ in $H_0$ that satisfies \eqref{eq0} in
  Claim \ref{cl:condit_connect} for the trivial event
  $I:= \{r_0 \sim_F r_0\}$ and $H':=H_0$. Let $r_1:=v$. Note that 
  $$
  \pr{r_1\sim_F\rt}= \pr{\deg_{F[H_0]}(r_1)=2} + \pr{\deg_{F[H_0]}(r_1)=1} 
  \le \frac{\alpha}{3} + \frac2n<1\,.
  $$
  The first inequality follows from \ref{edge_usage} and \eqref{eq0}, and the second
  since $\alpha\leq 2.5415$. Therefore $\pr{r_1\not\sim_F\rt}>0$.

  Suppose that we have defined $(H_0,r_1),(H_1,r_2),\ldots,(H_{i-1},r_i)$ for
  some $i$ with $1\le i\le k$, such that the following hold for all $1\leq j\leq i$: 
  (i) $r_j$ is a degree-2 node in $H_{j-1}$, and $r_{j-1}$ is the root of $H_{j-1}$;
  (ii) $\pr{r_{j-1}\not\sim_F\rt}>0$ if $j\geq 2$;
  (iii) if $j\geq 2$, then \eqref{eq0} and \eqref{eq1} hold in
  Claim~\ref{cl:condit_connect} for $H'=H_{j-1}$, $I=\{r_{j-1}\not\sim_F\rt\}$ and $v=r_j$.  
  We now show how to define $H_i$ and $r_{i+1}$ such that the above properties continue to
  hold for $j=i+1$. 

  First, set $H_i$ to be the copy of $H$ whose root is $r_i$. 
  We have $\pr{r_{i} \not \sim_F \rt}\ge \pr{r_{i-1}\not\sim_F \rt}>0$, so property (ii)
  continues to hold. 
  Given this, pick a degree-$2$ node $v$ in $H_i$ that satisfies \eqref{eq0} and \eqref{eq1} in 
  Claim~\ref{cl:condit_connect} for the event
  $I:= \{r_i \not \sim_F \rt\}$ and $H':=H_i$. Set $r_{i+1}:=v$. Thus, properties (i) and (iii)
  continue to hold as well.

\bigskip
For $j\in \{0,\ldots, k\}$, due to the choice of $r_{j+1}$ and \eqref{eq0}, we have
$\pr{\deg_{F[H_{j}]}(r_{j+1})=1} \le \frac{2}{n}$ 
and thus
\begin{equation}
  \label{eq3}
  \pr{r_{j+1}\sim_F r_j}= \pr{\deg_{F[H_j]}(r_{j+1})=2} + \pr{\deg_{F[H_j]}(r_{j+1})=1} 
  \le \frac{\alpha}{3} + \frac2n\,. 
\end{equation}
For $j\in \{1,\ldots, k\}$, due to \eqref{eq1} and the choice of $r_{j+1}$, we get
\begin{align}
\pr{Q_{r_{j+1}}\subseteq F \,\land\, r_j\not \sim_F r_0}
&=\pr{Q_{r_{j+1}}\subseteq F \,|\, r_j\not \sim_F r_0} \cdot\pr{r_j\not \sim_F r_0}\notag \\
&\ge\frac{2(n-1)}{3n}\pr{r_j\not\sim_F r_0}\notag \\
&\geq \frac{2}{3} \bigl(1-\pr{r_j\sim_F r_0} \bigr) - \frac{2}{3n}\,.\label{eq4}
\end{align}

Hence, for $j\in\{0,\ldots,k\}$,
\begin{align}
\pr{r_{j+1}\sim_F r_0} &=\pr{r_{j+1}\sim_F  r_j \,\land\, r_j\sim_F r_0} \notag \\
&=\pr{r_{j+1}\sim_F r_j} - \pr{r_{j+1}\sim_F r_j \,\land\, r_j\not \sim_F r_0} \notag \\
&\le \pr{r_{j+1}\sim_F r_j} - \pr{Q_{r_{j+1}}\subseteq F \,\land\, r_j\not \sim_F r_0}\notag \\
&\le \frac{\alpha}{3} + \frac{2}{n} -\frac{2}{3}+\frac{2}{3}\pr{r_j\sim_F r_0} + \frac{2}{3n}\,,
\label{eq5}
\end{align}
where  the first inequality follows from the fact that $r_{j+1}\sim_F r_j$ holds whenever $Q_{r_{j+1}}\subseteq F $ holds,
and the second inequality follows from \eqref{eq3} and \eqref{eq4}.

Expanding the recursion, we get 
\begin{align*}
\pr{r_{k+1}\sim_F r_0}\le \left( \frac{\alpha-2}{3} + \frac{8}{3n} \right) \sum_{i=0}^{k} \left(\frac{2}{3}\right)^{i} + \left(\frac{2}{3}\right)^{k+1},
\end{align*}
so for $n$ and $k$ large enough we obtain  
\begin{align*}
\pr{r_{k+1}\sim_F r_0}\le \left( \frac{\alpha-2}{3} + \frac{\epsilon}{6} \right) \sum_{i=0}^{\infty} \left(\frac{2}{3}\right)^{i} + \frac{\epsilon}{2} = \alpha -2+ \epsilon.
\end{align*}

Since $u:=r_{k+1}$ is a degree-$2$ node in $G$, the proof is complete.
\end{proof}

Now, we can prove the first part of Theorem~\ref{thm:pcsf-int-gap}. By Claim~\ref{cl:not_connected_node} and property~\ref{connectivity} of the collection of forests, we get the inequality
$$
\alpha -2 \ge 1-\alpha/3\,,
$$
leading to $\alpha\ge 9/4$.

\subsection{The Integrality Gap is Tight for the Construction}
We note that for any $n$ and $k$, the \PCSF instance given by our construction has integrality gap at most $9/4$. 
More generally, we show that the integrality gap over \PCSF instances which admit
a feasible solution $(x,z)$ to \eqref{lp} with $z_i\in\{0,1/3\}$ for all $i$, is at most $9/4$.
(That is, the maximum ratio between the optimal values of the IP and the LP for such 
instances is at most $9/4$.) This nicely complements our integrality-gap lower bound, and shows
that our analysis above is tight (for such instances).

\bigskip

To show the first statement, we simply provide a distribution over forests $F_1$,\ldots,
$F_q$ satisfying \ref{edge_usage} and \ref{connectivity}. (The next paragraph, which
proves the second claim above, gives another proof.) 
Since $(2(n-1)/(3n))\cdot \1$ is in the spanning tree polytope of $P$, 
there is a list of spanning trees such that every edge is contained in less than $2/3$ of them.
Consider the following distribution of forests. 
With probability $3-\alpha$ we pick one of these spanning trees of $P$ uniformly at random 
and subdivide it to obtain a tree in $H$; we take this tree in each copy of $H$ to obtain a (non-spanning) tree in $G$.
With probability $\alpha - 2$ we pick an arbitrary spanning tree of $G$.
This random forest $F$ satisfies
\begin{align*}
\pr{e \in F} \le (\alpha - 2) \cdot 1 + (3-\alpha) \cdot \frac{2}{3}= \frac{\alpha}{3}.
\end{align*}
Thus~\ref{edge_usage} holds for the above distribution. To see that~\ref{connectivity} holds, note that for every degree-$2$ node $v$ in $G$ we have  $\pr{v \sim_F \rt}\ge \alpha -2 = 1-\alpha/3$.

\bigskip

For the second claim, we utilize threshold rounding to show that the integrality gap is at
most $9/4$ for such instances. 
Consider an instance of $\PCSF$ and a feasible point $(x,z)$ for \eqref{lp} 
such that the values of $z$-variables are $0$ or $\gamma$ for some fixed $\gamma$ with $0<\gamma<1/2$. 
Using \cite{AKR,GW}, we can obtain an integer solution of cost at most $2c^\tp x+\pi^\tp z/\gamma$
by paying the penalties for all pairs with a non-zero $z$ value. 
We can also obtain a solution of cost at most $2c^\tp x/(1-\gamma)$ by connecting all pairs.
Therefore, for any $p\in[0,1]$, we can obtain an integer solution of cost at most
$$
p\left(2c^\tp x+\frac{\pi^\tp z}{\gamma}\right)+(1-p)\left(\frac{2c^\tp x}{1-\gamma}\right)
\leq\max \left\{ \frac{2-2p\gamma}{1-\gamma},\, \frac{p}{\gamma} \right\}(c^\tp x+\pi^\tp z)
$$
showing that the integrality gap is at most
$$
\mu:=\min_{0\le p\le 1} \max \left\{ \frac{2-2p\gamma}{1-\gamma},\, \frac{p}{\gamma} \right\}\,.
$$
The number $\mu$ is at most $2/(2\gamma^2-\gamma+1)$, which is equal to $9/4$ for $\gamma=1/3$. 
Note that for $\gamma=1/4$ the $2/(2\gamma^2-\gamma+1)$ achieves its maximum value of $16/7$.

\subsection{Lagrangian-Multiplier Preserving Approximation Algorithms for
  $\PCSF$} \label{lmp} 

Recall that a $\beta$-Lagrangian-multiplier-preserving (LMP) approximation algorithm
for $\PCSF$ is an approximation algorithm that returns a forest $F$ satisfying
\[ c(F) + \beta\cdot \pen{F} \leq \beta \cdot \opt\,. \]
We show that we must have $\beta\geq 4$ in order to obtain a $\beta$-LMP algorithm
relative to the optimum of the LP-relaxation \eqref{lp}, that is, to obtain the guarantee   
$c(F) + \beta\cdot \pen{F} \leq \beta \cdot \opt_{\text{\ref{lp}}}$. 
To obtain this lower bound, we modify our earlier construction slightly. We construct
$G=H^{(k)}$ in a similar fashion as before, but we now choose $P$ (the ``base graph'') to
be an $n$-node $l$-regular $l$-edge-connected graph. Let $x_e:=1/l$ for all
$e\in E$, and let $z_{uv}:=0$ if $u$ and $v$ are degree-$l$ nodes in
the same copy of $H$, and $z_{uv}:=1-2/l$ otherwise.

By arguments similar to~\cite{CV} (see, e.g., the proof of Theorem 7.2
in~\cite{GeorgiouS13}, and Theorem~\ref{lmp-theorem} in the Appendix), one can show that if
there exists a $\beta$-LMP approximation algorithm for \PCSF relative to \eqref{lp} then
there are forests $F_1$,\ldots, $F_q$ in $G$ (the same forest could appear multiple times)
such that picking a forest $F$ uniformly at random from $F_1,\ldots,F_q$
satisfies \smallskip 
\begin{enumerate}[(a')]
  \item \label{lmp_edge_usage} $\pr{e \in F} \leq \frac{\beta}{l}$ for all $e \in E$, and
  \item \label{lmp_connectivity} $\pr{u \sim_F v} \geq (1- z_{uv}) = \begin{cases} 
      1 & \mbox{$u$, $v$ are degree-$l$ nodes in the same copy of $H$} \\
      \frac{2}{l} & \mbox{if $u=\rt$ and $v$ is a degree-$2$ node in $G$}  \end{cases}$    
       
      for all $u,v \in V(G)$. 
\end{enumerate}

It is straightforward to obtain the analogues of Claim~\ref{cl:condit_connect} and Claim~\ref{cl:not_connected_node}.
\begin{claim}
Let a forest $F$ be picked uniformly at random from $F_1,\ldots,F_q$, let $I$ be an event with $\pr{I}>0$ and let $H'$ be a copy of $H$ in $G$. 
There exists a degree-$2$ node $v$ in $H'$, such that 
\begin{equation}
\pr{\deg_{F[H']}(v)=1} \le \frac{2}{n}
\end{equation}
and
\begin{equation}
\pr{Q_v\subseteq F \,|\, I}\ge\frac{2(n-1)}{l n}\,,
\end{equation}
where $Q_v$ is the path in $H'$ that contains $v$ and corresponds to an edge of $P$.
\end{claim}

\begin{claim} \label{cl:lmp_not_connected_node}
Let $\epsilon > 0$ be given.
Then for $n$ and $k$ sufficiently large, 
and choosing $F$ uniformly at random from $F_1,\ldots,F_q$, 
there exists a degree-$2$ node $u$ in $G$ such that
$$
\pr{u\sim_F \rt} \le \frac{\beta -2}{l-2}+\epsilon\,.
$$
\end{claim}

For the node $u$ from Claim~\ref{cl:lmp_not_connected_node}, we have $(\beta -2)/(l-2)\ge \pr{v\sim_F \rt} \ge 2/l$. 
Thus, $\beta$ is at least $4-4/l$, which approaches  $4$ as $l$ increases. This completes
the proof of the second part of Theorem~\ref{thm:pcsf-int-gap}. 

Moreover, the analysis is tight for the above construction. 
For a solution to \eqref{lp} where $z$ takes on only two distinct values,  say $0$ and $\gamma$, 
threshold rounding shows that for $\beta=2+2\gamma<4$ the desired collection of forests exists. 
However, for an unbounded number of distinct values of $z$, no constant-factor upper bound is known. 

\section{An Extreme Point for $\PCST$ with All Values at most \boldmath$\frac13$}

In this section we present a proof of Theorem \ref{thm:PCST}.  Take an
integer $k\ge 4$ and consider the graph $G=(V,E)$ in
Figure~\ref{fig:PCST_example}. Here, the nodes $v_1$,\ldots, $v_k$
represent the gadgets shown in Figure~\ref{fig:PCST_gadget}. The
gadget consists of ten nodes, and there are precisely four edges
incident to a node in the gadget. We let $r$ to be the root node
and introduce a source-sink node pair $(v,r)$ for every node
$v\in V\setminus\{r\}$.

\begin{figure}[ht]
\centering
\begin{tikzpicture}[inner sep = 2.5pt, yscale=.5, xscale=1.5]
\def\distEpsilon{0.9}
\tikzstyle{vtx}=[circle,draw,thick]
\tikzstyle{gdt}=[rectangle,draw,thick]
\tikzstyle{cut}=[line width=2.5pt]
\tikzstyle{anedge}=[thick] 
\node[vtx] (r) at (3,6) {$r$};
\node[gdt] (v1) at (0,0) {$v_1$};
\node[gdt] (v2) at (1,0) {$v_2$};
\node[gdt] (v3) at (2,0) {$v_3$};
\node[gdt] (v4) at (5,0) {$v_{k-1}$};
\node[gdt] (v5) at (6,0) {$v_{k}$};
\node (dots) at (3.5,0){$\cdots\cdots$};
\node[vtx] (s) at (3,-6) {$s$};
\draw[anedge] (s) [decorate]-- (v1);
\draw[anedge] (s) [decorate]-- (v2);
\draw[anedge] (s) [decorate]-- (v3);
\draw[anedge] (s) [decorate] -- (v4);
\draw[anedge] (s) [decorate] -- (v5);
\draw[anedge] (r) [decorate] -- (v1);
\draw[anedge] (r) [decorate] -- (v2);
\draw[anedge] (r) [decorate] -- (v3);
\draw[anedge] (r) [decorate] -- (v4);
\draw[anedge] (r) [decorate] -- (v5);
\draw[anedge] (v1) [decorate]-- (v2);  
\draw[anedge] (v2) [decorate]-- (v3);
\draw[anedge] (v3) [decorate]-- (3,0);      
\draw[anedge] (v4) [decorate]-- (4,0);      
\draw[anedge] (v4) [decorate]-- (v5);         
\draw[anedge] (v1) [decorate, bend right = 60] to (v5);        
\draw[cut, rounded corners=15pt, purple]
  ($(r)-(\distEpsilon,2*\distEpsilon)$) rectangle ($(r)+(\distEpsilon,\distEpsilon)$);
     
\end{tikzpicture}
\caption{Here, each of the nodes $v_1$,\ldots, $v_k$ corresponds to the gadget in Figure~\ref{fig:PCST_gadget}. 
Additionally, a cut $\{r\}$ is marked as a tight constraint in (PCST-LP) for the constructed point $(x,z)$. 
$x_e = 1/k$ for all edges $e$.}
\label{fig:PCST_example}
\end{figure}
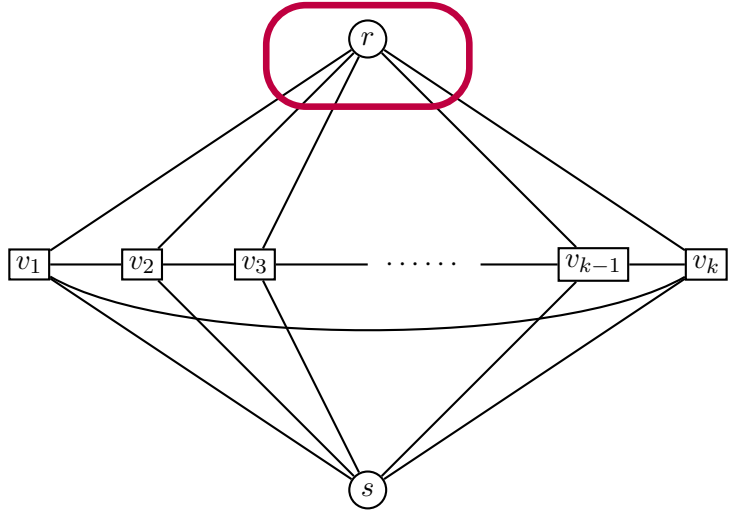

In the case $k=6$, the next claim proves Theorem~\ref{thm:PCST}.

\begin{claim}\label{claimpcstextremepoint}
The following is an extreme point of (PCST-LP) for this instance: 
$z_s=0$ and $z_u=1-4/k$ for every node $u$ in $V\setminus\{r,s\}$. 
For the wavy edges in Figure~\ref{fig:PCST_gadget}, we have $x_{u_1u_2}:=x_{u_3u_4}:=x_{u_5u_6}:=x_{u_7u_8}:=x_{u_9u_{10}}:=2/k$, 
and $x_e = 1/k$ for all the other edges $e$.
\end{claim}
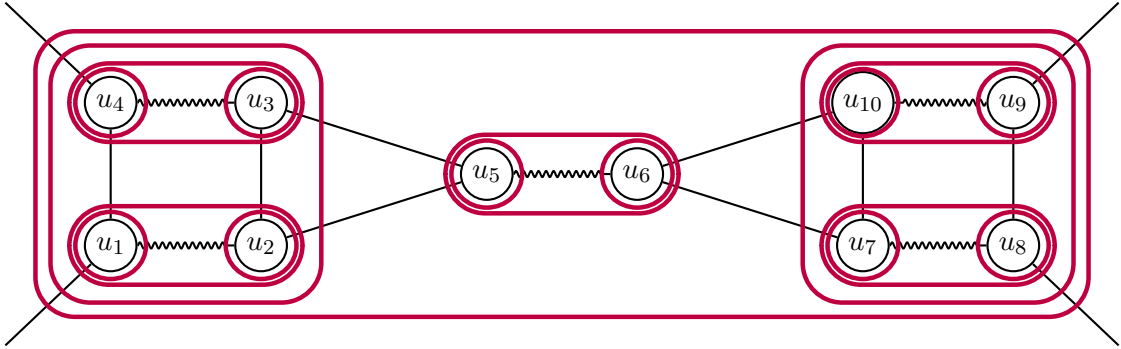
\begin{figure}[ht]
\centering
\begin{tikzpicture}[inner sep = 2.5pt, yscale=0.95]
\def\distExtra{1.4}
\def\distEpsilon{.55}
\def\distEpsilonII{.8}
\def\distEpsilonIII{1.}
\def\smallCutRadius{.47}
\tikzstyle{vtx}=[circle,draw,thick]
\tikzstyle{cut}=[line width=1.7pt, purple]
\tikzstyle{highweight}=[thick, decorate, decoration={snake, amplitude=.4mm, segment length=1mm, post length=1mm}]
\tikzstyle{lowweight}=[thick]
\node[vtx] (u1) at (0,0) {$u_1$};
\node[vtx] (u2) at (2,0) {$u_2$};
\node[vtx] (u3) at (2,2) {$u_3$};
\node[vtx] (u4) at (0,2) {$u_4$};
\node[vtx] (u5) at (5,1) {$u_5$};
\node[vtx] (u6) at (7,1) {$u_6$};
\node[vtx] (u7) at (10,0) {$u_7$};
\node[vtx] (u8) at (12,0) {$u_8$};
\node[vtx] (u9) at (12,2) {$u_9$};
\node[vtx] (u10) at (10,2) {$u_{10}$};
\draw[lowweight] (u2) [decorate]-- (u3);
\draw[highweight] (u1) -- (u2);
\draw[lowweight] (u4)[decorate] -- (u1);
\draw[highweight] (u4) -- (u3);
\draw[lowweight] (u1) [decorate]-- ($(u1)-(\distExtra,\distExtra)$);
\draw[lowweight] (u2) [decorate]-- (u5);
\draw[lowweight] (u3) [decorate]-- (u5);
\draw[lowweight] (u4) [decorate]-- ($(u4)-(\distExtra,-\distExtra)$);
\draw[highweight] (u5) -- (u6);
\draw[lowweight] (u8) [decorate]-- ($(u8)-(-\distExtra,\distExtra)$);   
\draw[lowweight] (u9) [decorate]-- ($(u9)+(\distExtra,\distExtra)$);      
\draw[lowweight] (u7) [decorate]-- (u10);          
\draw[lowweight] (u10) [decorate]-- (u6) [decorate]--(u7);      
\draw[lowweight] (u9) [decorate]-- (u8);    
\draw[highweight] (u9) -- (u10);     
\draw[highweight] (u7) -- (u8);                
\draw[cut] (u1) circle (\smallCutRadius);
\draw[cut] (u2) circle (\smallCutRadius);
\draw[cut] (u3) circle (\smallCutRadius);
\draw[cut] (u4) circle (\smallCutRadius);
\draw[cut] (u5) circle (\smallCutRadius);
\draw[cut] (u6) circle (\smallCutRadius);
\draw[cut] (u7) circle (\smallCutRadius);
\draw[cut] (u8) circle (\smallCutRadius);
\draw[cut] (u9) circle (\smallCutRadius);
\draw[cut] (u10) circle (\smallCutRadius);
\draw[cut, rounded corners=15pt]
  ($(u1)-(\distEpsilon,\distEpsilon)$) rectangle ($(u2)+(\distEpsilon,\distEpsilon)$);
\draw[cut, rounded corners=15pt]
  ($(u4)-(\distEpsilon,\distEpsilon)$) rectangle ($(u3)+(\distEpsilon,\distEpsilon)$);
 \draw[cut, rounded corners=15pt]
 ($(u1)-(\distEpsilonII,\distEpsilonII)$) rectangle ($(u3)+(\distEpsilonII,\distEpsilonII)$);
 \draw[cut, rounded corners=15pt]
  ($(u5)-(\distEpsilon,\distEpsilon)$) rectangle ($(u6)+(\distEpsilon,\distEpsilon)$);
\draw[cut, rounded corners=15pt]
  ($(u7)-(\distEpsilon,\distEpsilon)$) rectangle ($(u8)+(\distEpsilon,\distEpsilon)$);
\draw[cut, rounded corners=15pt]
  ($(u10)-(\distEpsilon,\distEpsilon)$) rectangle ($(u9)+(\distEpsilon,\distEpsilon)$);
 \draw[cut, rounded corners=15pt]
 ($(u7)-(\distEpsilonII,\distEpsilonII)$) rectangle ($(u9)+(\distEpsilonII,\distEpsilonII)$);  
 \draw[cut, rounded corners=15pt]
 ($(u1)-(\distEpsilonIII,\distEpsilonIII)$) rectangle ($(u9)+(\distEpsilonIII,\distEpsilonIII)$);    
  \end{tikzpicture}
\caption{A gadget used for the construction in Figure~\ref{fig:PCST_example}. 
Additionally, the cuts  are marked as tight constraints in (PCST-LP) for the constructed point $(x,z)$. For an edge $e$, $x_e = 2/k$ if $e$ is a wavy edge, and $x_e = 1/k$ if it is a straight edge.}
\label{fig:PCST_gadget}
\end{figure}
\begin{proof}
It is straightforward to check that the defined point $(x,z)$ is feasible. 
Let us show that the defined point $(x,z)$ is a vertex of (PCST-LP). 
To show this, it is enough to provide a set of tight constraints in (PCST-LP) which uniquely define the above point $(x,z)$.

Let us consider the gadget in Figure~\ref{fig:PCST_gadget}. 
For each such gadget, the set of tight inequalities from (PCST-LP) contains the following constraints:
\begin{align}
x(\delta(u_i))+z_{u_i} &= 1 ~~&&\forall i\in \{1,\ldots, 10\}\label{eq:PCST_1}\\
x(\delta(\{u_1,\ldots,u_{10}\}))+z_{u_i} &= 1 ~~&&\forall i\in \{1,\ldots, 10\}\label{eq:PCST_2}\\\
x(\delta(\{u_{i },u_{i+1}\}))+z_{u_i} &= 1 ~~&&\forall i\in \{1,3,5,7,9\}\label{eq:PCST_3}\\
x(\delta(\{u_{1},\ldots,u_{4}\}))+z_{u_1} &= 1\label{eq:PCST_4}\\
x(\delta(\{u_{7},\ldots,u_{10}\}))+z_{u_7} &= 1\label{eq:PCST_5}\,.
\end{align}
There are two more tight constraints which we use in the proof:
\begin{align}
x(\delta(r))+z_s &=1\label{eq:PCST_6}\\
z_s &=0\label{eq:PCST_7}\,.\hspace*{3.4cm}
\end{align}

Let us prove that the constraints~\eqref{eq:PCST_1}--\eqref{eq:PCST_7} define the point $(x,z)$ from Claim \ref{claimpcstextremepoint}. 
First, let us consider a gadget in Figure~\ref{fig:PCST_gadget}. 
It is clear that~\eqref{eq:PCST_2} implies $z_{u_1}=\ldots=z_{u_{10}}$. By~\eqref{eq:PCST_1} and~\eqref{eq:PCST_3}, we get
$$
2 x_{u_1u_2}=x(\delta(u_1))+x(\delta(u_2))-x(\delta(\{u_1,u_2\}))=(1-z_{u_1})+(1-z_{u_1})-(1-z_{u_1})=(1-z_{u_1})\,,
$$
and hence $x_{u_1u_2}=(1-z_{u_1})/2$. Similarly, we obtain $x_{u_1u_2}=x_{u_3u_4}=\ldots=x_{u_9u_{10}}=(1-z_{u_1})/2$. 

Now, we have
\begin{align*}
&x_{u_3 u_5}+x_{u_2 u_3}=x(\delta(u_3))-x_{u_3u_4}=(1-z_{u_1})/2\\
&x_{u_2 u_3}+x_{u_2 u_5}=x(\delta(u_2))-x_{u_1u_2}=(1-z_{u_1})/2\\\
&x_{u_2 u_5}+x_{u_3 u_5}=x(\delta(u_5))-x_{u_5u_6}=(1-z_{u_1})/2\,,
\end{align*}
implying $x_{u_2 u_3}=x_{u_2 u_5}=x_{u_3 u_5}=(1-z_{u_1})/4$. Similarly, $x_{u_6 u_7}=x_{u_6 u_{10}}=x_{u_7 u_{10}}=(1-z_{u_1})/4$.

By~\eqref{eq:PCST_3} and~\eqref{eq:PCST_4}, we get
$$
2 x_{u_1 u_4}=x(\delta(\{u_1,u_2\}))+x(\delta(\{u_3,u_4\}))-x(\delta(\{u_1,\ldots,u_4\}))-2x_{u_2 u_3}=(1-z_{u_1})/2\,,
$$
showing $x_{u_1 u_4}=(1-z_{u_1})/4$. Similarly, we get $x_{u_8 u_9}=(1-z_{u_1})/4$. 
From here, it is straightforward to show that all straight edges in Figure~\ref{fig:PCST_gadget} have value $(1-z_{u_1})/4$ and all wavy edges have value $(1-z_{u_1})/2$.

Consider the graph in Figure~\ref{fig:PCST_example}. 
Due to the edge $v_1v_2$, the straight edges in the gadget associated to $v_1$ have the same $x$ value as the straight edges in the gadget associated to $v_2$. 
Thus, due to the cycle $v_1v_2\ldots v_k$ the straight edges in all gadgets have the same $x$ value. 
To finish the proof use~\eqref{eq:PCST_6} and~\eqref{eq:PCST_7}.
\end{proof}

\section*{Acknowledgement}
We would like to thank Hausdorff Research Institute for Mathematics. 
This research was initiated during the Hausdorff Trimester Program ``Combinatorial Optimization".
NO was partially supported by an NWO Veni grant.
RR was supported in part by the U. S. National Science Foundation under award number
CCF-1527032. 
CS was supported in part by NSERC grant 327620-09 and an NSERC Discovery Accelerator
Supplement Award.  JK and KP were supported in part by NSERC Discovery Grant No 277224.

\bibliography{pcsf}
\bibliographystyle{plain}

\appendix

\section{Implications of an LMP Approximation Algorithm for \PCSF}
We adapt the arguments in~\cite{CV} to show that a $\beta$-LMP approximation relative to
\eqref{lp} implies that any fractional solution $(x,z)$ to \eqref{lp} can be translated to
a distribution over integral solutions to \eqref{lp} satisfying certain properties; this
implies the existence of the forests $F_1,\ldots,F_q$ in Section~\ref{lmp}. The arguments
below are known (see, e.g., the proof of Theorem 7.2 in~\cite{GeorgiouS13}); we include
them for completeness. 

Let $G=(V,E)$, $\{c_e\geq 0\}_{e\in E}$, $\{(s_i,t_i,\pi_i)\}_{i=1}^k$ be a \PCSF-instance.
Let $\{(x^{(q)},z^{(q)})\}_{q\in\I}$ be the set of all integral solutions to \eqref{lp},
where $\I$ is simply an index set.

\begin{theorem} \label{lmp-theorem}
Let $\A$ be a $\beta$-LMP approximation algorithm for \PCSF relative to \eqref{lp}. 
Given any fractional solution $(x^*,z^*)$ to \eqref{lp}, one can obtain 
nonnegative multipliers $\{\lambda^{(q)}\}_{q\in\I}$ such that $\sum_q\lambda^{(q)}=1$,  
$\sum_q\lambda^{(q)}x^{(q)}\leq\beta x^*$, and
$\sum_{q}\lambda^{(q)}z^{(q)}\leq z^*$. Moreover, the $\lambda^{(q)}$ values are rational
if $(x^*,z^*)$ is rational.
\end{theorem}

\begin{proof}
Consider the following pair of primal and dual LPs.

\vspace{-3ex}

{\centering\small
\noindent \hspace*{-6ex}
\begin{minipage}[t]{.475\textwidth}
\begin{alignat*}{2}
\max & \quad & \sum_q\lambda^{(q)} & \tag{P} \label{primal} \\
\text{s.t.} && \sum_q\lambda^{(q)} x_e^{(q)} & \leq \beta x^*_e \qquad \forall e \\ 
&& \sum_{q}\lambda^{(q)}z_{i}^{(q)} & \leq z^*_i  \qquad \forall i\\ 
&& \sum_q\lambda^{(q)} & \leq 1 \\ 
&& \lambda & \geq 0. \notag
\end{alignat*}
\end{minipage}
\quad \rule[-31ex]{1pt}{28ex}\ 
\begin{minipage}[t]{.52\textwidth}
\begin{alignat*}{2}
\min & \quad & \sum_e \beta x^*_ed_e+\sum_iz^*_i\rho_i & + \gamma \tag{D} \label{dual} \\ 
\text{s.t.} && \sum_e x_e^{(q)}d_e+\sum_{i}z_{i}^{(q)}\rho_i + \gamma & \geq 1 \qquad \forall q 
\\ 
&& d,\rho,\gamma & \geq 0. \notag
\end{alignat*}
\end{minipage}
}

\smallskip
It suffices to show that the optimal value of \eqref{primal} is 1. The rationality of the
$\lambda^{(q)}$ values when $(x^*,z^*)$ is rational then follows from the fact that an LP
with rational data has a rational optimal solution. (The proof below also yields a
polynomial-time algorithm to solve \eqref{primal} by showing that $\A$ can be used to
obtain a separation oracle for the dual.) 

Note that both \eqref{primal} and \eqref{dual} are feasible, so they have a common
optimal value. We show that $\opt_{\ref{dual}}=1$. Setting $\gamma=1$, $d=\rho=\0$,
we have that $\opt_{\ref{dual}}\leq 1$. Suppose $(d,\rho,\gamma)$ is feasible to
\eqref{dual} and $\sum_e \beta x^*_ed_e+\sum_iz^*_i\rho_i+\gamma<1$. Consider the \PCSF
instance given by $G$, edge costs $\{d_e\}_{e\in E}$, and terminal pairs and
penalties $\{(s_i,t_i,\rho_i/\beta)\}_{i=1}^k$. Running $\A$ on this instance, we can
obtain an integral solution $(x^{(q)},z^{(q)})$ such that
$$
\sum_ed_e x^{(q)}_e+\sum_i\rho_iz^{(q)}_i + \gamma \leq
\beta\Bigl(\sum_ed_e x^*_e+\sum_iz^*_i\rho_i/\beta \Bigr) + \gamma <1
$$
which contradicts the feasibility of $(d,\rho,\gamma)$. Hence, $\opt_{\ref{dual}}=1$.
\end{proof}

Note that if $(x^*,z^*)$ is rational, then since the $\lambda^{(q)}$ values are rational,
we can multiply them by a suitably large number to convert them to integers; thus, we may
view the distribution specified by the $\lambda^{(q)}$ values as the 
{\em uniform distribution} over a {\em multiset} of integral solutions to \eqref{lp}.

We remark that the converse of Theorem~\ref{lmp-theorem} also holds in the following
sense. If for every fractional solution $(x^*,z^*)$ to \eqref{lp}, we can obtain
$\lambda^{(q)}$ values (or equivalently, a distribution over integral solutions to
\eqref{lp}) satisfying the properties in Theorem~\ref{lmp-theorem}, then we can obtain a
$\beta$-LMP approximation algorithm for \PCSF relative to \eqref{lp}: this follows, by
simply returning the integral solution $(x^{(q)},z^{(q)})$ with $\lambda^{(q)}>0$ that
minimizes $\sum_ec_ex^{(q)}_e+\beta\sum_i\pi_iz^{(q)}_i$.

\end{document}